\documentclass[copyright,creativecommons]{eptcs}
\providecommand{\event}{GANDALF 2014} 

	\usepackage{amsmath,amsthm,amsfonts}
	\usepackage{xspace}
	\usepackage{graphicx}
	\usepackage{tikz}
	\usepackage[small]{subfigure}
	\usepackage{paralist}
	\usepackage{dsfont}
	\usepackage{breakurl}             
	\usepackage{xfrac}
	\usepackage{bbm}

	\usepackage[ruled,linesnumbered,noline]{algorithm2e}
	\SetKw{Let}{let}
	\SetKwInOut{Input}{Input}

\newcommand{\bbI}{\mathbb{I}}
\def\bbR{\mathbb{R}} 
\def\bbN{\mathbb{N}} 
\def\bbNpos{\mathbb{N}_{\geq 1}} 

\newcommand{\One}{\mathbbmss{1}}
\newcommand{\Half}{\mathbbmss{t}}

\newcommand{\floor}[1]{{\lfloor {#1} \rfloor}}

\newtheorem{definition}{Definition}
\newtheorem{theorem}{Theorem}

\newtheorem{lemma}{Lemma}
\newtheorem{proposition}{Proposition}

\DeclareMathOperator{\Untime}{Untime}

\newcommand{\dims}[1]{k}

\newcommand{\trans}[1]{\xrightarrow{#1}}

\newcommand{\cut}[1]{ }

	\title{Deterministic Timed Finite State Machines: Equivalence Checking and Expressive Power}
	\author{Davide Bresolin
		\institute{University of Bologna \\ 
		Bologna, Italy}
		\email{davide.bresolin@unibo.it}
	\and 
		Khaled El-Fakih
		\institute{American University of Sharjah \\
		Sharjah, United Arab Emirates}
		\email{kelfakih@aus.edu}
	\and 
		Tiziano Villa
		\institute{University of Verona \\ 
		Verona, Italy}
		\email{tiziano.villa@univr.it}
	\and 
		Nina Yevtushenko
		\institute{Tomsk State University \\
		Tomsk, Russia}
		\email{yevtushenko@sibmail.com}}

	\def\titlerunning{Timed Finite State Machines: Models and Properties}
	\def\authorrunning{D. Bresolin, K. El-Fakih, T. Villa \& N. Yevtushenko}

\begin{document}

	\maketitle

	\begin{abstract}
	There has been a growing interest in defining models of automata enriched with time.
	For instance, timed automata were introduced as automata extended with clocks.
	In this paper, we study models of timed finite state machines (TFSMs), i.e.,
	FSMs enriched with time, which accept timed input words and generate timed 
	output words.
	Here we discuss some models of TFSMs with a single clock: 
	TFSMs with timed guards, TFSMs with timeouts, and TFSMs with both timed 
	guards and timeouts.
	We solve the problem of equivalence checking for all three models, 
	and we compare their expressive power, characterizing subclasses of TFSMs 
	with timed guards and of TFSMs with timeouts that are equivalent to each
	other.
	\end{abstract}

	\section{Introduction}

Finite automata (FA) and finite state machines (FSMs) are formal models 
widely used in the practice of engineering and science, e.g.,
in application domains ranging from sequential circuits, communication 
protocols, embedded and reactive systems, to biological modelling. 

Since the 90s, the standard classes of FA have been enriched with 
the introduction of time constraints to represent more accurately 
the behaviour of systems in discrete or continuous time.
Timed automata (TA) are such an example: they are finite automata augmented 
with a number of resettable real-time clocks, whose transitions 
are triggered by predicates involving clock values~\cite{Alur-tcs1994}.

More recently, timed models of FSMs have been proposed in the literature
by the introduction of time constraints such as timed guards or timeouts.
Timed guards restrict the input/output transitions to happen within given
time intervals. In particular, the timed FSM proposed in~\cite{Gromov2009, ElFakih2013, ElFakih-scp2014} features: one clock variable, time constraints 
to limit the time elapsed at a state, and clock reset when a transition
is executed.

The timed FSM proposed in~\cite{Merayo2008,Hierons-jlap2009} 
features: one clock variable, time constraints to limit the time elapsed 
when an output has to be produced after an input has been applied to the FSM, 
clock reset when an output is produced, timeouts. 
The meaning of timeouts is the following: if no input is applied at 
a current state for some timeout period, the timed FSM moves from 
the current state to another state using a timeout function;
e.g., timeouts are common in telecommunication protocols and systems. 

TA and TFSMs are also used when deriving tests for discrete event systems.
However, methods for deriving complete finite test suites with a
guaranteed fault coverage exist only for TFSMs, therefore TFSMs are
preferred over TA and other models, when the derivation of complete tests
is required.

%
%

In this paper, we investigate some models of TFSMs with a single clock: 
TFSMs with only timed guards, TFSMs with only timeouts, and TFSMs with 
both timed guards and timeouts.
We solve the problem of equivalence checking for all three models, 
we compare their expressive power, and we characterize subclasses of TFSMs 
with timed guards and of TFSMs with timeouts that are equivalent to each
other.
These results are obtained by introducing relations of bisimulation that 
define untimed finite state machines whose states include information
on the clock regions. This is reminiscent of the region graph construction
used to prove that in timed automata the verification questions 
(e.g., expressed by safety properties) have the same answer for all the clock valuations in the same clock region~\cite{Alur-tcs1994}.
In our case, we are able to prove a stronger result: the timed behaviours of
two timed FSMs are equivalent if and only if the behaviours of the companion
untimed FSMs are equivalent.
So our models of timed FSMs strike a good balance
between expressivity and computational complexity.



	\section{Timed FSM Models}
        \label{sec:timedfsms}

	Let $A$ be a finite alphabet, and let $\bbR^+$ be the set of non-negative reals. A \emph{timed symbol} is a pair $(a,t)$ where $t \in \bbR^+$ is called the \emph{timestamp} of the symbol $a\in A$. 
	A timed word is then defined as a finite sequence $(a_1,t_1)(a_2,t_2)(a_3,t_3)\dots$ of timed symbols where the sequence of timestamps $t_1 t_2 t_3 \dots$ is increasing.
	All the timed models considered in this paper are input/output machines that operate by reading a \emph{timed input word} $(i_1,t_1)(i_2,t_2)\dots(i_k,t_k)$ defined on some \emph{input alphabet $I$}, and producing a corresponding \emph{timed output word} $(o_1,t_1)$ $(o_2,t_2)\dots(o_k,t_k)$ on some \emph{output alphabet $O$}. 
	The production of outputs is assumed to be instantaneous: the timestamp of the $j$-th output $o_j$ is the same of the $j$-th input $i_j$. Models where there is a delay between reading an input and producing the related output are possible but not considered in this paper.
	Given a timed word $(a_1,t_1)(a_2,t_2)\dots(a_k,t_k)$, 
	$\Untime((a_1,t_1)(a_2,t_2)\dots(a_k,t_k)) = a_1 a_2 \dots a_k$ denotes
	the word obtained when deleting the timestamps.

	A timed possibly non-deterministic and partial FSM (TFSM) is an FSM 
	augmented with a clock. The clock is a real number that measures the time delay
	at a state, and its value is reset to zero when a transition is executed.
	In this section we first introduce the TFSM model with timed guards given in~\cite{ElFakih2013,Gromov2009} and the TFSM model with timeouts given in~\cite{Merayo2008,Zhigulin2011}. Then, we define a TFSM model with both timed guards and timeouts that subsumes the other two. In addition, we study the equivalence problem for each of the three TFSM models.

	\subsection{TFSM with timed guards}\label{sec:TFSM-TG}

	A timed guard defines the time interval when a transition can be executed.
	Intuitively, a TFSM in the present state $s$ and accepting input $i$ at a time 
	$t$ satisfying the timed guard responds with output $o$ and moves to the next 
	state $s'$, while the clock is reset to 0 and restarts advancing in state $s'$.

	\begin{definition}[TFSM with Timed Guards~\cite{ElFakih2013,Gromov2009}]\label{def:fsm-tg}
	A TFSM with timed guards is a tuple $M = (S, I, O, \lambda_S, s_0)$
	where $S$, $I$, and $O$ are finite disjoint non-empty sets of states, inputs and outputs, respectively, $s_0$ is the initial state, $\lambda_S \subseteq  S \times I \times \Pi \times  O \times  S$ 
	is a transition relation where $\Pi$ is the set of input timed guards. Each guard in $\Pi$ is an interval $g = \langle t_{min}, t_{max} \rangle$ where $t_{min}$ is a nonnegative integer, while $t_{max}$ is either a nonnegative integer or $\infty$, $t_{min} \leq t_{max}$, and  $\langle \in \big\{ ( , [ \big\}$ while $\rangle \in \big\{ ), ] \big\}$. 
	\end{definition}

	The \emph{timed state} of a TFSM is a pair $(s,x)$ such that $s \in S$ is a state of $M$ and $x \in \bbR^+$ is the current value of the clock. \emph{Transitions} between timed states can be of two types:
	\begin{compactitem}
		\item \emph{timed transitions} of the form $(s,x) \trans{t} (s,x+t)$ where $t \in \bbR^+$, representing the fact that a delay of $t$ time units has elapsed without receiving any input;
		\item \emph{input/output transitions} of the form $(s,x) \trans{i,o} (s',0)$, representing reception of the input symbol $i \in I$, production of the output $o \in O$ and reset of the clock. An input/output transition can be activated only if there exists a tuple $(s,i,\langle t_{min},t_{max}\rangle,o,s') \in \lambda_S$ such that $x \in \langle t_{min},t_{max}\rangle$.
	\end{compactitem}


	A \emph{timed run} of a TFSM with timed guards $M$ interleaves timed transitions with input/output transitions.
	Given a timed input word $v = (i_1,t_1)(i_2,t_2)\dots$ $(i_k,t_k)$, a timed run of $M$ over $v$ is a finite sequence $\rho = (s_0,0) \trans{t_1} (s_0,t_1) \trans{i_1,o_1} (s_1,0) \trans{t_2-t_1} (s_1,t_2-t_1) \trans{i_2,o_2} (s_2,0) \trans{t_3-t_2} \dots  \trans{i_k,o_k} (s_k,0)$ such that $s_0$ is the initial state of $M$, and for every $j \geq 0$ $(s_j,0) \trans{t_{j+1}-t_j} (s_j,t_{j+1}-t_j) \trans{i_{j+1},o_{j+1}} (s_{j+1},0)$ is a valid sequence of transitions of $M$. The timed run $\rho$ is said to \emph{accept} the timed input word $v= (i_1,t_1)(i_2,t_2)\dots(i_k,t_k)$ and to \emph{produce} the timed output word $u = (o_1,t_1)(o_2,t_2)\dots(o_k,t_k)$. 	The behavior of $M$ is defined in terms of the input/output words accepted and produced by the machine.

	\begin{definition}\label{def:behavior}
	The behavior of a TFSM $M$ is a partial mapping $B_M: (I\times \bbR)^* \mapsto 2^{(O\times \bbR)^*}$ that associates every input word $w$ accepted by $M$ with the set of output words $B_M(w)$ produced by $M$ under input $w$. When $M$ is an untimed FSM the behavior is defined as a partial mapping $B_M: I^* \mapsto 2^{O^*}$.

Two machines $M$ and $M'$ with the same input and output alphabets are \emph{equivalent} if and only if they have same behavior, i.e, $B_M = B_{M'}$.
	\end{definition}

	\paragraph{Complete and deterministic machines.}

	The usual definitions for FSMs of deterministic and non-deterministic, submachine, etc., can be extended to all timed FSMs models considered here. In particular, a TFSM is {\em complete} if for each state $s$, input $i$ and value of the clock $x$ there exists at least one transition $(s,x) \trans{i,o} (s',0)$,  otherwise the machine is {\em partial}. A TFSM is \emph{deterministic} if for each state $s$, input $i$ and value of the clock $x$  there exists at most one input/output transition, otherwise is {\em non-deterministics}.

	For the sake of simplicity, from now on we consider only \emph{complete and deterministic} machines, leaving the treatment of partial and non-deterministic TFSM to future work.
	When a machine $M$ is deterministic and complete, we have that $B_M(w)$ is a singleton set for every input word $w$. Hence, we can redefine the behavior $B_M$ as a function $B_M: (I\times \bbR)^* \mapsto (O\times \bbR)^*$ that associates every input word $w = (i_1,t_1)(i_2,t_2)\dots(o_k,t_k)$ with the unique output word $B_M(w)$ produced by $M$ under input $w$.

	Moreover, we can consider the transition relation of the machine as a complete function $\lambda_S: S \times I \times \bbR^+ \mapsto S \times O$ that takes as input the current state $s$, the delay $t$ and the input symbol  $i$ and produces the (unique) next state and output symbol $\lambda_S(s,t,i) = (s',o)$ such that $(s,0)\trans{t}(s,t)\trans{i,o}(s',0)$. With a slight abuse of the notation, we can extend it to a function $\lambda_S: S \times (I \times \bbR^+)^* \mapsto S \times O^*$ that takes as inputs the initial state $s$ and a timed word $w$, and returns the state reached by the machine after reading $w$ and the generated output word. We will use $s \trans{w,u} s'$ as a shorthand for $\lambda_S(s,w) = (s',u)$.

	\paragraph{Equivalence checking of TFSM with timed guards.}

	 In this section we show how to solve  the equivalence problem of TFSM with guards by reducing it to the equivalence problem of untimed FSM. We proceed in three steps: first, we show how to build an ``abstract'' FSM from a TFSM with guards; then we define an appropriate notion of bisimulation to compare TFSM with guards with untimed FSM; finally, from the properties of the bisimulation relation, we conclude that two TFSM with guards are equivalent if and only if their abstractions are equivalent.

	Now, let $M$ be a TFSM with guards. We define $\max(M)$ as the greatest integer constant (different from $\infty$) appearing in the guards of $\lambda_S$. For any natural number $N \geq \max(M)$, we define $\bbI_N$ as the set of intervals 
	%
$	\bbI_N = \{[n,n] \mid n \leq N\} \cup \{(n,n+1) \mid 0 \leq n < N\} \cup \{(N,\infty)\}.$
The discrete abstraction of a TFSM with guards will take as inputs pairs of the form $(i,\langle n,n'\rangle)$ where $i$ is the actual input and $\langle n,n'\rangle \in \bbI_N$ an interval representing the time delay.

	\begin{definition}\label{def:abstract-tfsm-tg}
	Given a TFSM with timed guards $M = (S, I, O, \lambda_S, s_0)$ and a natural number $N \geq \max(M)$, we define the \emph{abstract FSM} $A^N_M = (S, I\times \bbI_N, O, \lambda_A, s_0)$ as the untimed FSM such that $(s, (i,\langle n,n'\rangle), o, s') \in \lambda_A$ if and only if $(s, i, \langle t,t'\rangle, o, s') \in \lambda_S$ for some guard $\langle t,t'\rangle$ such that $\langle n,n'\rangle \subseteq \langle t,t'\rangle$.
	\end{definition}

	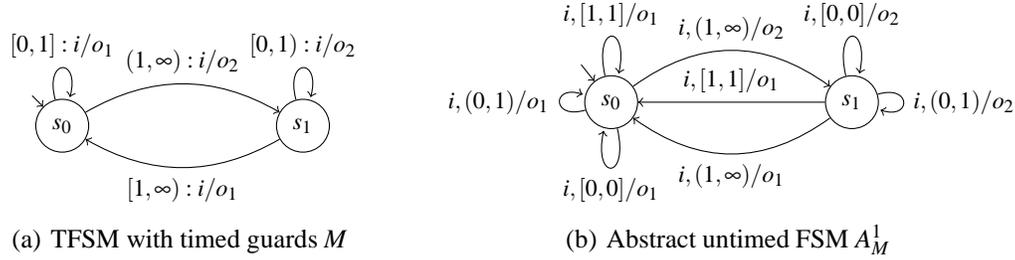
\begin{figure}[tbp]
	\centering
		\subfigure[TFSM with timed guards $M$]{
		\begin{tikzpicture}[font=\footnotesize,scale=0.8]
			\node[draw,circle,minimum size=20pt] (q0)	 at (0,0)	{$s_0$};
			\node[draw,circle,minimum size=20pt]	(q1) at (4,0)	{$s_1$};
			
			\draw[->] (-0.5,0.5) -- (q0);
			\path[->] (q0) edge[loop above] node        {$[0,1]:i/o_1$} ();
			\path[->] (q0) edge[bend left]  node[above] {$(1,\infty):i/o_2$} (q1);
			\path[->] (q1) edge[bend left]  node[below] {$[1,\infty):i/o_1$} (q0);
			\path[->] (q1) edge[loop above] node        {$[0,1):i/o_2$} ();
		\end{tikzpicture}\label{fig:tfsm-tg-ex}}
	\qquad 
		\subfigure[Abstract untimed FSM $A^1_M$]{
		\begin{tikzpicture}[font=\footnotesize,xscale=0.8]
			\node[draw,circle,minimum size=20pt] (q0)	 at (0,0)	{$s_0$};
			\node[draw,circle,minimum size=20pt]	(q1) at (4,0)	{$s_1$};
			
			\draw[->] (-0.5,0.5) -- (q0);
			\path[->] (q0) edge[loop below] node        {$i,[0,0]/o_1$} ();
			\path[->] (q0) edge[loop left] node        {$i,(0,1)/o_1$} ();
			\path[->] (q0) edge[loop above] node        {$i,[1,1]/o_1$} ();
			\path[->] (q0) edge[bend left]  node[above] {$i,(1,\infty)/o_2$} (q1);
			\path[->] (q1) edge node[above] {$i,[1,1]/o_1$} (q0);
			\path[->] (q1) edge[bend left] node[below] {$i,(1,\infty)/o_1$} (q0);
			\path[->] (q1) edge[loop above] node        {$i,[0,0]/o_2$} ();
			\path[->] (q1) edge[loop right] node        {$i,(0,1)/o_2$} ();
				\end{tikzpicture}\label{fig:tfsm-tg-ab}}

		\caption{{$\bbI_N$-abstraction} of TFSM with timed guards.}
	  \label{fig:tfsm-tg-abstraction}
	\end{figure}

	Figure~\ref{fig:tfsm-tg-abstraction} shows an example of a simple TFSM with timed guards and of the corresponding {$\bbI_N$-abstraction} (for $N = 1$).
	We cannot directly compare the behavior of a timed FSM with the behavior of its untimed abstraction, since the former accepts timed input words on $I$ and the latter accepts untimed input words on $I \times \bbI_N$. For this reason, we need to introduce the notion of abstraction of a timed word. 

	\begin{definition}\label{def:abstract-timed-word}
	Given a finite alphabet $A$, a finite timed word $v = (a_1,t_1)$ $(a_2,t_2)(a_3,t_3)\dots(a_m,t_m)$, an integer $N$ and the set of intervals $\bbI_N$, we define its \emph{$\bbI_N$-abstraction} as the finite word 
	$\bbI_N(v) = (a_1,\langle n_1,n_1'\rangle)$ $(a_2,\langle n_2,n_2'\rangle)(a_3,\langle n_3,n_3'\rangle)\ldots(a_m,\langle n_m,n_m'\rangle)$
	such that $t_i - t_{i-1} \in \langle n_i,n_i'\rangle$ for every $1 \leq i \leq m$.
	\end{definition}

	 $\bbI_N$-bisimulation connects states of a timed FSM with states of an untimed FSM.

	\begin{definition}\label{def:I_n-bisim}
	Given a TFSM with timed guards $T = (S, I, O, \lambda_S, s_0)$, an untimed FSM $U = (R, I\times \bbI_N, O, \lambda_R, r_0)$, an integer $N \geq max(T)$ and the set of intervals $\bbI_N$, an \emph{$\bbI_N$-bisimulation} is a relation $\sim \subseteq S \times R$ that respects the following conditions:
	\begin{compactenum}
		\item for every pair of states $s \in S$ and $r \in R$ such that $s \sim r$, if $(s, i, \langle t,t'\rangle,o,s') \in \lambda_S$, then for every $\langle n,n'\rangle \in \bbI_N$ such that $\langle n,n'\rangle \subseteq \langle t,t'\rangle$ there exists a transition $(r, (i,\langle n,n'\rangle), o, r') \in \lambda_R$ such that $r' \sim s'$;
		\item for every pair of states $s \in S$ and $r \in R$ such that $s \sim r$, if $(r, (i,\langle n,n'\rangle), o, r') \in \lambda_R$ then there exists a transition $(s, i, \langle t,t'\rangle,o,s') \in \lambda_S$ such that $\langle n,n'\rangle \subseteq \langle t,t'\rangle$ and $r' \sim s'$.
	\end{compactenum}
	$T$ and $U$ are  \emph{$\bbI_N$-bisimilar} if there exists an $\bbI_N$-bisimulation $\sim \subseteq S \times R$ such that $s_0 \sim r_0$.
	\end{definition}

	$\bbI_N$-bisimilar machines have the same behavior, as formally proved by the following lemma.

	\begin{lemma}\label{lem:I_N-bisimilar-then-equiv}
	Given a TFSM with timed guards $T = (S, I, O, \lambda_S, s_0)$ and an untimed FSM $U = (R, I\times \bbI_N, O, \lambda_R, r_0)$, if there exists an $\bbI_N$-bisimulation $\sim$ such that $s_0 \sim r_0$ then for every timed input word $v = (i_1,t_1)(i_2,t_2)(i_3,t_3)\dots(i_m,t_m)$ we have that $\Untime(B_T(v)) = B_U(\bbI_N(v))$.
	\end{lemma}

	\begin{proof}
	We prove the lemma by showing that 
\emph{``for every pair of states $s \sim r$ and timed word $v$, $\lambda_S(s,v) = (s',w)$ if and only if $\lambda_R(r,\bbI_N(v)) = (r',\Untime(w))$ with $s' \sim r'$''}.

	We prove the claim by induction on the length $m$ of the input word $v$. Suppose $m = 1, v = (i_1,t_1)$ and $\lambda_S(s, (i_1,t_1)) = (s_1,(o_1,t_1)) = (s_1,w)$. By the definition of TFSM with timed guards, we have that there exists a transition $(s, i, \langle t,t'\rangle, o_1, s_1) \in \lambda_S$ such that $t_1 \in \langle t,t'\rangle$. 
	Since $s \sim r$, by the definition of $\bbI_N$-bisimulation we have that there exists a transition $(r, (i,\langle n_1,n_1'\rangle),o_1,r_1) \in \lambda_R)$ with $t_1 \in \langle n_1,n_1'\rangle$ and $s_1 \sim r_1$. Since $\bbI_N(v) = (i_1,\langle n_1,n_1'\rangle)$, we have that $\lambda_R(r,(i_1,\langle n_1,n_1'\rangle)) = (r_1,o_1)$ and the claim is proved.

	Now, suppose the claim holds for all natural numbers up to $m-1$, and let $v = (i_1,t_1)\dots(i_{m-1},t_{m-1})$ $(i_m,t_m) = v'(i_m,t_m)$. Suppose that $\lambda_S(s,v') = (s_{m-1},w')$ and that $\lambda_S(s_{m-1},(i_m,t_m-t_{m-1})) = (s_m,(o_m,t_m-t_{m-1}))$: by inductive hypothesis, we have that $\lambda_R(r,\bbI_N(v')) = (r_{m-1},\Untime(w'))$ and that $\lambda_R(r_{m-1},$ $\bbI_N(i_m,t_m-t_{m-1})) = (r_m,o_m)$ with $s_{m-1} \sim r_{m-1}$ and $s_m \sim r_m$. This implies that $\lambda_S(s,v) =$ $\lambda_S(s,v'(i_m,t_m)) = (s_m,w'(o_m,t_m))$ $= (s_m, w)$ and $\lambda_R(r,\bbI_N(v)) = \lambda_R(r,\bbI_N(v'(i_m,t_m))) = (r_m,\Untime(w')o_m) =$\linebreak $(r_m,\Untime(w))$ with $s_m \sim r_m$, and thus that the claim holds also for $m$.

	To conclude the proof of the Lemma it is sufficient to recall that from the definition of behaviour we have that $B_T(v) = w$ if and only if $\lambda_S(s_0,v) = (s_m,w)$ for some state $s_m \in S$. From $s_0 \sim r_0$ we can conclude that $\lambda_R(r_0,\bbI_N(v)) = (r_m,\Untime(w))$ and thus that $B_U(\bbI_N(v)) = \Untime(w) = \Untime(B_T(v))$.
	\end{proof}

	\begin{lemma}\label{lem:abstract-I_N-bisim}
	A TFSM with timed guards $M$ is $\bbI_N$-bisimilar to the abstract FSM $A^N_M$.
	\end{lemma}

	\begin{proof}
	The identity relation on $S$ is an $\bbI_N$-bisimulation for $M$ and $A^N_M$.
	\end{proof}

	\begin{theorem}\label{th:tfsm-tg-equiv}
	Let $M$ and $M'$ be two TFSM with timed guards and let $N$ be such that $N \geq \max(M)$ and $N \geq \max(M')$.
	Then $M$ and $M'$ are equivalent if and only if the two abstract FSM $A^N_M$ and $A^N_{M'}$ are equivalent.
	\end{theorem}

	\begin{proof}
	By Lemma~\ref{lem:abstract-I_N-bisim} we have that $M$ is $\bbI_N$-bisimilar to $A^N_M$ and that $M'$ is $\bbI_N$-bisimilar to $A^N_{M'}$. We first show that if $M$ and $M'$ are equivalent, then $A^N_M$ and $A^N_{M'}$ are equivalent. Under the assumption that $M$ and $M'$ are complete and deterministic, we have that $M$ and $M'$ are equivalent if for every timed input word $v$, $B_M(v) = B_{M'}(v)$. By Lemma~\ref{lem:I_N-bisimilar-then-equiv} we have that $\Untime(B_M(v)) = B_{A^N_M}(\bbI(v))$ and that $\Untime(B_{M'}(v)) = B_{A^N_{M'}}(\bbI(v))$. This implies that $B_{A^N_M}(\bbI(v)) = B_{A^N_{M'}}(\bbI(v))$. Since for every untimed word $w$ over the alphabet $I \times \bbI_N$ it is possible to find a timed word $v$ over the alphabet $I$ such that $w = \bbI_N(v)$, we can conclude that $A^N_M$ and $A^N_{M'}$ are equivalent.

	For the converse implication, suppose by contradiction that $A^N_M$ and $A^N_{M'}$ are equivalent but $M$ and $M'$ are not. This means that there exists a timed input word $v = (i_1,t_1)\ldots(i_m,t_m)$ such that $B_M(v) \neq B_{M'}(v)$. Let $B_M(v) = w = (o_1,t_1)\ldots(o_m,t_m)$ and $B_{M'}(v) = w' = (o_1',t_1)\ldots(o_m',t_m)$.
	Notice that, by the definition of TFSM with timed guards, $w$ and $w'$ must have the same timestamps (outputs are produced istantaneuously), and thus they must differ on the produced output symbols. This implies that there exists at least one index $1 \leq j \leq m$ such that $o_j \neq o_j'$ and hence that $\Untime(w) \neq \Untime(w')$.  By Lemma~\ref{lem:I_N-bisimilar-then-equiv} we have that $B_{A^N_M}(\bbI(v)) = \Untime(w)$ and that $B_{A^N_{M'}}(\bbI(v)) = \Untime(w')$. Hence, $B_{A^N_M}(\bbI(v)) \neq B_{A^N_{M'}}(\bbI(v))$, in contradiction with the hypothesis that $A^N_M$ and $A^N_{M'}$ are equivalent.
	\end{proof}

	\subsection{TFSM with timeouts}\label{sec:TFSM-TO}

	\begin{definition}[TFSM with Timeouts~\cite{Merayo2008,Zhigulin2011}]\label{def:fsm-to}
	A TFSM with timeouts  is a 6-tuple $M = (S, I, O, \lambda_S, s_0, \Delta_S)$
	where $S$, $I$, and $O$ are finite disjoint non-empty sets of states, inputs and outputs, respectively, $s_0$ is the initial state, $\lambda_S \subseteq  S \times I \times  O \times  S$ is a transition relation and  
	$\Delta_S: S \rightarrow S \times \left(\bbNpos \cup \left\{\infty\right\} \right)$ is a \textit{timeout function}.
	\end{definition}

	Transitions of TFSM with timeouts can be triggered not only by the reception of an input, but also by timeouts.
	When the machine enters a state $s$ it resets the clock to $0$. If an input $i$ is received before the timeout $\Delta_S\left(s\right)_{\downarrow \bbN}$ expires and a transition $\left(s, i, o, s' \right) \in$  $S \times I \times  O \times  S$ exists, then the machine produces $o$, moves to state $s'$ while resetting the clock at $s'$ to $0$. If no input is received before the timeout  $\Delta_S\left(s\right)_{\downarrow \bbN}$ expires, then the TFSM will move to the state specified by the timeout function $\Delta_S\left(s\right)_{\downarrow S}$ and reset the clock to $0$.
	If $\Delta_S\left(s\right)_{\downarrow \bbN} = \infty$, then the machine can stay at state $s$ infinitely long waiting for an input. 

	A timed state of a TFSM with timeouts is a pair $(s,x) \in S \times \bbR^+$ with the additional constraint that $x < \Delta_S(s)_{\downarrow \bbN}$ (the value of the clock cannot exceed the timeout). Timed and input/output transitions are defined as follows.

	\begin{compactitem}
		\item The \emph{timed transition} relation $\trans{t}$ is the smallest relation closed under the following properties:
			\begin{compactitem}
				\item for every timed state $(s,x)$ and delay $t \geq 0$, if $x + t < \Delta_S(s)_{\downarrow \bbN}$, then $(s,x) \trans{t} (s, x+t)$;
				\item for every timed state $(s,x)$ and delay $t \geq 0$, if $x + t = \Delta_S(s)_{\downarrow \bbN}$, then $(s,x) \trans{t} (s', 0)$ with $s' = \Delta_S(s)_{\downarrow S}$;
				\item if $(s,x) \trans{t_1} (s',x')$ and $(s',x') \trans{t_2} (s'',x'')$ then $(s,x) \trans{t_1+t_2} (s'',x'')$.
				\end{compactitem}

		\item The \emph{input/output transition} relation $\trans{i,o}$ is such that $(s,x) \trans{i,o} (s',0)$ if and only if $(s,i,o,s') \in \lambda_S$ and $x < \Delta_S(s)_{\downarrow \bbN}$.
	\end{compactitem}

	\noindent The definition of timed run, behavior, complete and deterministic machine given in Section~\ref{sec:TFSM-TG} for TFSM with timed guards can be extended to TFSM with timeouts.


	\paragraph{Equivalence checking of TFSM with timeouts.}

	We solve the equivalence problem for TFSM with timeouts using the same approach we used in Section~\ref{sec:TFSM-TG} for TFSM with timed guards: we reduce the problem to the equivalence of standard FSM by an appropriate notion of ``abstract'' untimed FSM. 

	In the case of a TFSM with timeouts $M$, the constant $\max(M)$ is defined as the greatest timeout value of the function $\Delta_S$ different from $\infty$. States of the abstract FSM will be pairs $(s,n)$ where $s$ is a state of $M$ and $n$ is a natural number ranging from $0$ to $\max(M)-1$ abstracting the clock value. Transitions can be either standard input/output transitions labelled with pairs from $I \times O$ or ``time elapsing'' transitions labelled with the special pair $(\One,\One)$ representing a one time-unit delay without inputs (see also~\cite{Zhigulin2011}).

	\begin{definition}\label{def:abstract-tfsm-to}
	Given a TFSM with timeouts $M = (S, I, O, \lambda_S, s_0, \Delta_S)$, let $N = \max(M) - 1$. We define the \emph{$\One$-abstract FSM} $A_M = (S\times \{0,\dots,N\}, I \cup \{\One\}, O\cup \{\One\}, \lambda_A, (s_0,0))$ as the untimed FSM such that:
	\begin{compactitem}
		\item $(s,n) \trans{\One,\One} (s,n+1)$ if and only if $n+1 < \Delta_S(s)_{\downarrow\bbN} < \infty$;
		\item $(s,n) \trans{\One,\One} (s',0)$ if and only if $\Delta_S(s) = (s',n+1)$;
		\item $(s,0) \trans{\One,\One} (s,0)$ if and only if $\Delta_S(s)_{\downarrow\bbN} = \infty$;
		\item $(s,n) \trans{i,o} (s',0)$ if and only if $(s,i,o,s') \in \lambda_S$.
	\end{compactitem}
	\end{definition}

	\begin{figure}[tbp]
	\centering
	\subfigure[FSM with timeouts]{
	\begin{tikzpicture}[font=\footnotesize,yscale=0.8]
		\node	(M1) at (-1.25,0) {$M$};
		\node[draw,circle] (q0)	 at (0,0)	{$q_0$};
		\node[draw,circle]	(q1) at (2,0)	{$q_1$};
		
		\draw[->] (-0.75,0) -- (q0);
		\path[->] (q0) edge[loop above] node        {$i/o_1$} ();
		\path[->] (q0) edge[bend left]  node[above] {$t=3$} (q1);
		\path[->] (q1) edge[bend left]  node[below] {$t=2$} (q0);
		\path[->] (q1) edge[loop above] node        {$i/o_2$} ();
		\path (0,-2) -- (2,2);
	\end{tikzpicture}\label{fig:One-tfsm-to}}
\qquad
	\subfigure[$\One$-abstract FSM]{
	\begin{tikzpicture}[font=\footnotesize,yscale=0.8]
		\node	(M1) at (-1.25,0) {$A_M$};
		\node[draw,circle] (q00)	 at (0,0)	{$q_0,0$};
		\node[draw,circle] (q01)	 at (2,0)	{$q_0,1$};
		\node[draw,circle] (q02)	 at (4,0)	{$q_0,2$};
		\node[draw,circle]	(q10) at (3,-1.75)	{$q_1,0$};
		\node[draw,circle]	(q11) at (1,-1.75)	{$q_1,1$};
		
		\draw[->] (-0.75,0) -- (q00);
		\path[->] (q00) edge[loop above] node        {$i/o_1$} ();
		\path[->] (q00) edge  node[below] {$\One/\One$} (q01);
		\path[->] (q01) edge[bend right]  node[above] {$i/o_1$} (q00);
		\path[->] (q01) edge  node[below] {$\One/\One$} (q02);
		\path[->] (q02) edge[bend right=60]  node[above] {$i/o_1$} (q00);
		\path[->] (q02) edge  node[right] {$\One/\One$} (q10);
		\path[->] (q10) edge[loop right] node        {$i/o_2$} ();
		\path[->] (q10) edge  node[below] {$\One/\One$} (q11);
		\path[->] (q11) edge[bend left]  node[above] {$i/o_2$} (q10);
		\path[->] (q11) edge  node[left] {$\One/\One$} (q00);
	\end{tikzpicture}\label{fig:One-fsm}}
	\caption{Example of $\One$-abstraction of a TFSM with timeouts.}
  \label{fig:One-abstraction-example}
\end{figure}
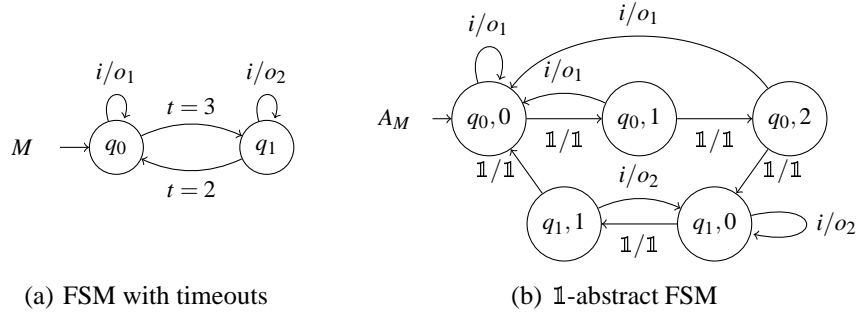

Figure~\ref{fig:One-abstraction-example} shows an example of a TFSM with timeouts and its $\One$-abstraction.
In this case the untimed abstraction accepts untimed input words on $I \cup \{\One\}$ where the delay is implicitly represented by sequences of the special input symbol $\One$ interleaved with the occurrences of the original input symbols from $I$. A sequence of $n$ $\One$ in a row represents a delay $t$ included in the right-open interval $[n,n+1)$. 

To compare timed words with untimed ones, we need to introduce the following notion of abstraction of a timed word. Given a real number $t \in \bbR$, we denote with $\floor{t}$ the integer part of $t$, and with $\One^{\floor{t}}$ a (possibly empty) sequence of $\floor{t}$ delay symbols $\One$.

\begin{definition}\label{def:One-abstract-timed-word}
Given a finite alphabet $A$ and a finite timed word $v = (a_1,t_1)$ $(a_2,t_2)(a_3,t_3)\dots(a_m,t_m)$, we define its \emph{$\One$-abstraction} as the finite word  $\One(v) = \One^{\floor{t_1}}a_1 \One^{\floor{t_2-t_1}} \dots \One^{\floor{t_j-t_{j-1}}} a_j \One^{\floor{t_{j+1}-t_j}}\dots \One^{\floor{t_m-t_{m-1}}} a_m$.
\end{definition}

$\One$-bisimulation connects \emph{timed states $(s,x)$} of a timed FSM with states of an untimed FSM.

\begin{definition}\label{def:One-bisim}
Given a TFSM with timeouts $T = (S, I, O, \lambda_S, s_0, \Delta_S)$ and an untimed FSM $U = (R, I \cup \{\One\}, O \cup \{\One\}, \lambda_R, r_0)$, a \emph{$\One$-bisimulation} is a relation $\sim \subseteq (S\times \bbR^+) \times R$ that respects the following conditions for every pair of states $(s,x) \in S\times \bbR^+$ and $r \in R$ such that $(s,x) \sim r$:
\begin{compactenum}
	\item if $(s,x)\trans{t}(s',x')$ with $\floor{x+t} = \floor{x+1}$ then there exists $r'\in R$ such that $r \trans{\One,\One} r'$ and $(s',x') \sim r'$;
	\item if $r \trans{\One,\One} r'$ then for every $t$ such that $\floor{x+t} = \floor{x+1}$ there exists $(s',x')\in S\times \bbR^+$ such that $(s,x)\trans{t}(s',x')$ and $(s',x') \sim r'$;
	\item if $(s,x)\trans{t}(s,x+t)\trans{i,o}(s',0)$ with $\floor{x} = \floor{x+t}$ then there exists $r'\in R$ such that $r \trans{i,o} r'$ and $(s',0) \sim r'$;
	\item if $r \trans{i,o} r'$ then for every $t \in \bbR$ such that $\floor{x} = \floor{x+t}$ there exists $(s',0)\in S\times \bbR^+$ such that $(s,x)\trans{t}(s,x+t)\trans{i,o}(s',0)$ and $(s',0) \sim r'$.
\end{compactenum}
$T$ and $U$ are  \emph{$\One$-bisimilar} if there exists a $\One$-bisimulation $\sim \subseteq S \times R$ such that $(s_0,0) \sim r_0$.
\end{definition}
To understand the previous conditions, notice that $\floor{x+t} = \floor{x+1}$ 
implies $1 \leq t < 2$ and that $\floor{x} = \floor{x+t}$ implies $0 \leq t < 1$.
Therefore condition 1. refers to a timed transition in $T$ of length
$1 \leq t < 2$, which corresponds to the existence of a 
$\One,\One$-transition in $U$.
Similarly, condition 3. refers in $T$ to a timed transition of length
$0 \leq t < 1$ followed by an input-output transition, which corresponds
to an input-output transition in $U$.
Finally, condition 2. refers to a $\One,\One$-transition in $U$, which
corresponds to timed transitions in $T$ with $1 \leq t < 2$;
condition 4. refers to an input-output transition in $U$, which corresponds
in $U$ to timed transitions with $0 \leq t < 1$ followed by an input-output 
transition.
The timed transitions for $t \geq 2$ are handled by induction in the
following Lemma~\ref{lem:One-bisimilar-time}.

\begin{lemma}\label{lem:One-bisimilar-time}
Given a TFSM with timeouts $T = (S, I, O, \lambda_S, s_0, \Delta_S)$ and an untimed FSM $U = (R, I\cup\{\One\}, O\cup\{\One\}, \lambda_R, r_0)$, every $\One$-bisimulation relation $\sim \subset (S\times\bbR^+)\times R$ respects the following properties for every $(s,0) \sim r$ and $t \geq 1$:
\begin{compactenum}[\it (i)]
	\item if $(s,0) \trans{t} (s',x')$ then there exists $r'$ such that $(s',x') \sim r'$ and $r \trans{\One,\One^{\floor{t}}} r'$;
	\item if $r \trans{\One,\One^{\floor{t}}} r'$  then there exists $(s',x') \sim r'$ such that $(s,0) \trans{t} (s',x')$.
\end{compactenum}
\end{lemma}

\begin{proof}
The proof is by induction on $\floor{t}$.
For the basis of the induction, suppose $\floor{t} = 1$ ($1 \leq t < 2$) and let $(s,0) \sim r$. The two properties are a direct consequence of the definition of $\One$-bisimulation. By condition 1 of Definition~\ref{def:One-bisim}, we have that for every $1 \leq t < 2$, $(s,0) \trans{t} (s',x')$ implies that there exists $r'$ such that $(s',x') \sim r'$ and $r \trans{\One,\One} r'$. By condition 2 of Definition~\ref{def:One-bisim}, we have that for every $1 \leq t < 2$, $r \trans{\One,\One} r'$ implies that there exists $(s',x') \sim r'$ such that $(s,0) \trans{t} (s',x')$.

For the inductive case, suppose that $\floor{t} \geq 2$ and that the Lemma holds for $\floor{t - 1} \geq 1$. Now, let $(s,0) \trans{t} (s',x')$ and consider the timed state $(s'',x'')$ such that $(s,0) \trans{t-1} (s'',x'') \trans{1} (s',x')$. By inductive hypothesis we have that there exists $r''$ such that $(s'',x'') \sim r''$ and $r \trans{\One,\One^{\floor{t-1}}} r''$. By condition 1 of Definition~\ref{def:One-bisim}, we have that there exists $r'$ such that $(s',x') \sim r'$ and $r'' \trans{\One,\One} r'$ and thus that $r \trans{\One,\One^{\floor{t}}} r'$. To prove property \textit{(ii)}, suppose $r \trans{\One,\One^{\floor{t}}} r'$ and consider the state $r''$ such that $r \trans{\One,\One^{\floor{t-1}}} r''\trans{\One,\One} r'$. By inductive hypothesis we have that there exists $(s'',x'') \sim r''$ such that $(s,0) \trans{t-1} (s'',x'')$. By condition 2 of Definition~\ref{def:One-bisim} it is possible to find a state $(s',x')$ such that $(s',x') \sim r'$ and $(s'',x'') \trans{1} (s',x')$. This shows that $(s,0) \trans{t} (s',x')$ and concludes the proof.
\end{proof}

$\One$-bisimilar machines have the same behavior, as formally proved in the following.

\begin{lemma}\label{lem:One-bisimilar-then-equiv}
Given a TFSM with timeouts $T = (S, I, O, \lambda_S, s_0, \Delta_S)$ and an untimed FSM $U = (R, I\cup\{\One\}, O\cup\{\One\}, \lambda_R, r_0)$, if there exists a $\One$-bisimulation $\sim$ such that $(s_0,0) \sim r_0$ then for every timed input word $v = (i_1,t_1)\dots(i_m,t_m)$ we have that $\One(B_T(v)) = B_U(\One(v))$.
\end{lemma}

\begin{proof}
We prove the lemma by showing that 
\emph{
`` for every pair of states $s \in S$ and $r \in R$ such that $(s,0) \sim r$ and timed word $v$, $\lambda_S(s,v) = (s',w)$ if and only if $\lambda_R(r,\One(v)) = (r',\One(w))$ with $(s',0) \sim r'$''}.

We prove the claim by induction on the length $m$ of the input word. Suppose $m = 1$, $v = (i_1,t_1)$ and $w = (o_1,t_1)$. By the definition of TFSM with timeouts we have that $\lambda_S(s, (i_1,t_1)) = (s_1,(o_1,t_1))$ if and only if there exists a timed state $(s',x')$ such that $(s,0) \trans{t_1} (s',x') \trans{i_1,o_1} (s_1,0)$.
To prove the direct implication, suppose $\lambda_S(s,v)=(s',w)$.
We distinguish between two cases depending on the value of $t_1$.

\begin{compactitem}
	\item If $t_1 < 1$, by condition 3. of the definition of $\One$-bisimulation, there exists $r_1 \in R$ such that $r \trans{i_1,o_1} r_1$. Hence, $\lambda_R(r,\One(i_1,t_1)) = (r_1,\One(o_1,t_1))$.
	
	\item If $t_1 \geq 1$, by Lemma~\ref{lem:One-bisimilar-time} (i), there exists $r'$ such that $r \trans{\One,\One^\floor{t_1}} r'$ and $(s',x') \sim r'$. By condition 3. of the definition of $\One$-bisimulation, we have that it is possible to find a state $r_1 \in R$ such that $r' \trans{i_1,o_1} r_1$. This implies that under input $\One^\floor{t_1} i_1 = \One(i_1,t_1)$ the FSM $U$ produces the output word $\One^\floor{t_1} o_1 = \One(o_1,t_1)$, and thus we can conclude that $\lambda_R(r,\One(i_1,t_1)) = (r_1,\One(o_1,t_1))$.
\end{compactitem}

To prove the converse implication, suppose $\lambda_R(r,\One(i_1,t_1)) = (r_1,o_1) = (r_1,\One(o_1,t_1))$. We distinguish between two cases depending on the value of $t_1$.
\begin{compactitem}
	\item If $t_1 < 1$, by condition 4. of the definition of $\One$-bisimulation (applied with $t = 0$), there exists $(s_1,0) \in S \times \bbR$ such that $(s,0) \trans{t_1} (s,t_1) \trans{i_1,o_1} (s_1,0)$. Hence, $\lambda_S(s, (i_1,t_1)) = (s_1,(o_1,t_1))$.	
	\item If $t_1 \geq 1$, then by the assumption $\lambda_R(r,\One(i_1,t_1)) = (r_1,o_1) = (r_1,\One(o_1,t_1))$ there exists $r' \in R$ such that $r \trans{\One,\One^\floor{t_1}} r' \trans{i_1,o_1} r_1$. By Lemma~\ref{lem:One-bisimilar-time}(ii), there exists $(s',x')\in S \times \bbR$ such that $(s,0) \trans{t_1} (s',x')$ and $(s',x') \sim r'$. By condition 4. of the definition of $\One$-bisimulation (applied with $t = 0$), we have that there exists a timed state $(s_1,0)$ such that $(s',x') \trans{i_1,o_1} (s_1,0)$. This implies that under input $(i_1,t_1)$ the TFSM $T$ produces the timed output word $(o_1,t_1)$, and thus we can conclude that $\lambda_S(s, (i_1,t_1)) = (s_1,(o_1,t_1))$.	
\end{compactitem}

\medskip

To prove the inductive case, suppose $m > 1$, $v = (i_1,t_1)\ldots(i_m,t_m)$ and $w = (i_1,t_1)\ldots(i_m,t_m)$. Now, let $v'=(i_1,t_1)\ldots(i_{m-1},t_{m-1})$ and $w' = (o_1,t_1)\ldots$ $(o_{m-1},t_{m-1})$. By inductive hypothesis, we have that $\lambda_S(s,v') = (s_{m-1},w')$ if and only if $\lambda_R(r,\One(v')) = (r_{m-1},\One(w'))$ and that $\lambda_S(s_{m-1},(i_m,t_m-t_{m-1})) = (s_m,(o_m,t_m-t_{m-1}))$ if and only if $\lambda_R(r_{m-1},\One(i_m,t_m-t_{m-1})) = (r_m,\One(o_m,t_m-t_{m-1}))$ for some $(s_{m-1},0) \sim r_{m-1}$ and $(s_m,0) \sim r_m$. This implies that $\lambda_S(s,v'(i_m,t_m)) = (s_m,w'(o_m,t_m))$ if and only if $\lambda_R(r,\One(v)) = \lambda_R(r,\One(v'(i_m,t_m))) = (r_m,\One(w')\One(o_m,t_m-t_{m-1})) = (r_m,\One(w))$, and thus that the claim holds also for $m$.

To conclude the proof of the Lemma it is sufficient to recall that from the definition of behaviour we have that $B_T(v) = w$ if and only if $\lambda_S(s_0,v) = (s_m,w)$ for some state $s_m \in S$. From $(s_0,0) \sim r_0$ we can conclude that $\lambda_R(r_0,\One(v)) = (r_m,\One(w))$ and thus that $B_U(\One(v)) = \One(w) = \One(B_T(v))$.
\end{proof}

\begin{lemma}\label{lem:abstract-One-bisim}
A TFSM with timeouts $M$ is $\One$-bisimilar to the abstract FSM $A_M$.
\end{lemma}

\begin{proof}
The relation $\sim = \{((s,x),(s,n))\mid \Delta_S(s)_{\downarrow\bbN} < \infty$ and $\floor{x} = n\} \cup \{((s,x),(s,0)) \mid \Delta_S(s)_{\downarrow\bbN} = \infty\}$ is a $\One$-bisimulation for $M$ and $A_M$.
\end{proof}

\begin{theorem}\label{th:tfsm-tos-equiv}
Let $M$ and $M'$ be two TFSM with timeouts.
Then $M$ and $M'$ are equivalent if and only if the two abstract FSM $A_M$ and $A_{M'}$ are equivalent.
\end{theorem}


\subsection{TFSM with timeouts and timed guards}\label{sec:TFSM-TGTO}

In this paper we define a new timed FSM model that incorporates both guards and timeouts. 
Informally, each state of the machine has a timeout (possibly $\infty$) and all outgoing transitions of the state have timed guards with upper bounds less than the state timeout. 
As in the other models described above, time is set to zero when executing a transition. 
 
\begin{definition}[Timed FSM]\label{def:tfsm-all}
A timed FSM $\cal S$ is a finite state machine augmented with timed guards and timeouts. Formally, a timed FSM (TFSM) is a 6-tuple $(S, I, O, \lambda_S, s_0, \Delta_S)$
where $S$, $I$, and $O$ are finite disjoint non-empty sets of states, inputs and outputs, respectively, $s_0$ is the initial state,
$\lambda_S \subseteq  S \times \left(I \times \Pi \right) \times  O \times  S$ 
is a transition relation where $\Pi$ is the set of input timed guards, and
$\Delta_S: S \rightarrow S \times \left( N \cup \left\{ \infty \right\} \right)$
is a \textit{timeout function}.
\end{definition}

Similarly to FSMs with timeouts, if no input is applied at a current state $s$ before the timeout 
$\Delta_S\left(s\right)_{\downarrow \bbN}$ expires, then the TFSM will move to anther state 
$\Delta_S\left(s\right)_{\downarrow S}$ as prescribed by the timeout function. 
If $\Delta_S\left(s\right)_{\downarrow \bbN} = \infty$, then the TFSM can stay at state $s$ infinitely long waiting for an input. Similarly to FSMs with timed guards, an input/output transition can be triggered only if the value of the clock is inside the guard $\langle t_{min}, t_{max}\rangle$ labeling the transition.  
Timed transitions are thus defined as for TFSM with timeouts, while input/output transitions are defined as for TFSM with timed guards:

\begin{compactitem}
	\item The \emph{timed transition} relation $\trans{t}$ is the smallest relation closed under the following properties:
		\begin{compactitem}
			\item for every timed state $(s,x)$ and delay $t \geq 0$, if $x + t < \Delta_S(s)_{\downarrow \bbN}$, then $(s,x) \trans{t} (s, x+t)$;
			\item for every timed state $(s,x)$ and delay $t \geq 0$, if $x + t = \Delta_S(s)_{\downarrow \bbN}$, then $(s,x) \trans{t} (s', 0)$ with $s' = \Delta_S(s)_{\downarrow S}$;
			\item if $(s,x) \trans{t_1} (s',x')$ and $(s',x') \trans{t_2} (s'',x'')$ then $(s,x) \trans{t_1+t_2} (s'',x'')$.
			\end{compactitem}

	\item The \emph{input/output transition} relation $\trans{i,o}$ is such that $(s,x) \trans{i,o} (s',0)$ if and only if there exists $(s,i,\langle t_{min}, t_{max}\rangle,o,s') \in \lambda_S$ such that $x \in \langle t_{min}, t_{max}\rangle$.
\end{compactitem}

\paragraph{Equivalence checking of TFSM with timeouts and timed guards.}

We can solve the equivalence problem for TFSM with timeouts and timed guards by combining the techniques used for TFSM with timeouts (Section~\ref{sec:TFSM-TO}) and for TFSM with timed guards (Section~\ref{sec:TFSM-TG}). 
To incorporate the effect of guards in the abstract untimed FSM, we have to use a finer granularity for the ``time elapsing'' transitions, which are now labelled with the special pair $(\Half,\Half)$, which intuitively represents a time delay $0 < t^* < 1$ without inputs. In the case of a TFSM with timeouts and timed guards $M$, the constant $\max(M)$ is defined as the maximum between the greatest timeout value of the function $\Delta_S$ (different from $\infty$) and the greatest integer constant (different from $\infty$) appearing in the guards of $\lambda_S$. States of the abstract FSM will be pairs $(s,\langle n,n'\rangle)$ where $s$ is a state of $M$ and $\langle n,n'\rangle$ is either a point-interval $[n,n]$ or an open interval $(n,n+1)$ from the set 
$\mathbb{I}_N$ defined in~\ref{sec:TFSM-TG} for the abstraction of TFSM with timed guards.

\begin{definition}\label{def:abstract-tfsm-all}
Given a TFSM with timeouts and timed guards $M = (S, I, O, \lambda_S, s_0, \Delta_S)$, let $N = \max(M)$. We define the \emph{$\Half$-abstract FSM} $A_M = (S\times \mathbb{I}_N, I \cup \{\Half\}, O\cup \{\Half\}, \lambda_A, (s_0,[0,0]))$ as the untimed FSM such that:
\begin{compactitem}
	\item $(s,[n,n]) \trans{\Half,\Half} (s,(n,n+1))$ if and only if $n+1 \leq \Delta_S(s)_{\downarrow\bbN}$;
	\item $(s,(n,n+1)) \trans{\Half,\Half} (s,[n+1,n+1])$ if and only if $n+1 < \Delta_S(s)_{\downarrow\bbN}$;
	\item $(s,(n,n+1)) \trans{\Half,\Half} (s',[0,0])$ if and only if $\Delta_S(s) = (s',n+1)$;
	\item $(s,(N,\infty)) \trans{\Half,\Half} (s,(N,\infty))$ if and only if $\Delta_S(s)_{\downarrow\bbN} = \infty$;
	\item $(s,\langle n,n'\rangle) \trans{i,o} (s',[0,0])$ if and only if there exists $(s,i,\langle t,t'\rangle,o,s') \in \lambda_S$ such that $\langle n,n'\rangle \subseteq \langle t,t'\rangle$.
\end{compactitem}
\end{definition}

\begin{figure}[tbp]
\centering
	\subfigure[TFSM with timeouts and timed guards $M$]{
	\begin{tikzpicture}[font=\footnotesize,scale=0.8]
\node[draw,circle,minimum size=20pt] (q0)	 at (0,0)	{$s_0$};
\node[draw,circle,minimum size=20pt]	(q1) at (3.5,0)	{$s_1$};
		
\draw[->] (-0.5,0.5) -- (q0);
\path[->] (q0) edge[loop above] node        {$[0,1):i/o_1$} ();
\path[->] (q0) edge[bend left]  node[above] {$t = 1$} (q1);
\path[->] (q1) edge[bend left]  node[below] {$(1,\infty):i/o_1$} (q0);
\path[->] (q1) edge[loop above] node        {$[0,1]:i/o_2$} ();
\path (-2,-2.25)--(6,2.25);
	\end{tikzpicture}\label{fig:tfsm-all-ex}}
\qquad
	\subfigure[$\Half$-Abstract untimed FSM $A_M$]{
\begin{tikzpicture}[yscale=0.75,xscale=1.1,font=\footnotesize]
		\node[draw,circle,minimum size=20pt,text width=4ex,text centered] (q00)	 at (-4,2.25)	{$s_0$ \\ $[0,0]$};
		\node[draw,circle,minimum size=20pt,text width=4ex,text centered] (q005)	 at (-2,2.25)	{$s_0$ \\ $(0,1)$};
		\node[draw,circle,minimum size=20pt,text width=4ex,text centered]	(q10) at (0,2.25)	{$s_1$ \\ $[0,0]$};
		\node[draw,circle,minimum size=20pt,text width=4ex,text centered]	(q105) at (1,0)	{$s_1$ \\ $(0,1)$};
		\node[draw,circle,minimum size=20pt,text width=4ex,text centered]	(q11) at (-1,0)	{$s_1$ \\ $[1,1]$};
		\node[draw,circle,minimum size=20pt,text width=4ex,text centered]	(q115) at (-3,0)	{$s_1$ \\ $(1,\infty)$};
		
		\draw[->] (-4.5,3.25) -- (q00);
		\path[->] (q00) edge[loop above] node        {$i/o_1$} ();
		\path[->] (q00) edge node[below]  {$\Half/\Half$} (q005);
		\path[->] (q005) edge[bend right] node[above]        {$i/o_1$} (q00);
		\path[->] (q005) edge node[above] {$\Half/\Half$} (q10);
		\path[->] (q10) edge[loop above] node  {$i/o_2$} ();
		\path[->] (q10) edge node[left] {$\Half/\Half$} (q105);
		\path[->] (q105) edge[bend right] node[right]  {$i/o_2$} (q10);
		\path[->] (q105) edge node[below] {$\Half/\Half$} (q11);
		\path[->] (q11) edge[bend left] node[near start,left]  {$i/o_2$} (q10);
		\path[->] (q11) edge node[below] {$\Half/\Half$} (q115);
		\path[->] (q115) edge[bend left] node[left] {$i/o_1$} (q00);
		\path[->] (q115) edge[loop left] node {$\Half/\Half$} (q115);
			\end{tikzpicture}\label{fig:tfsm-all-ab}}

	\caption{{$\Half$-abstraction} of TFSM with timeout and timed guards.}
  \label{fig:tfsm-all-abstraction}
\end{figure}

Figure~\ref{fig:tfsm-all-abstraction} shows an example of a TFSM with timeouts and its $\Half$-abstraction. In this case the untimed abstraction accepts untimed input words on $I \cup \{\Half\}$. As for TFSM with timeouts, the delay is implicitly represented by sequences of the special input symbol $\Half$ interleaving the occurrences of the real input symbols from $I$. In this case the representation of delays is more involved:
\begin{compactitem}
	\item an \emph{even number} $2n$ of $\Half$ symbols represents a delay of \emph{exactly} $n$ time units;
	\item an \emph{odd number} $2n + 1$ of $\Half$ symbols represents a delay $t$ included in the open interval $(n,n+1)$.
\end{compactitem}
The notion of abstraction of a timed word captures the above intuition. 

\begin{definition}\label{def:Half-abstract-timed-word}
Let $\Half(t)$ be a function mapping a delay $t \in \bbR$ to a sequence of $\Half$ as follows:
$\Half(t) = 	\Half^{2t}$ if $\floor{t} = t$,  $\Half(t) = \Half^{2\floor{t} + 1}$ otherwise.
Given a finite alphabet $A$ and a finite timed word $v = (a_1,t_1)$ $(a_2,t_2)(a_3,t_3)\dots(a_m,t_m)$, we define its \emph{$\Half$-abstraction} as the finite word  $\Half(v) = \Half(t_1)a_1 \Half(t_2-t_1) \dots \Half(t_j-t_{j-1}) a_j \Half(t_{j+1}-t_j)\dots \Half(t_{m-1}-t_m) a_m$.
\end{definition}

The definition of $\Half$-bisimulation is similar to the one of $\One$-bisimulation. As in Definition~\ref{def:One-bisim}, conditions 1. and 2. formalize the connection between timed transitions and the special symbol $\Half$. The finer granularity of the time delays allows us to simplify conditions 3. and 4.: differently from Definition~\ref{def:One-bisim}, we do not need to consider timed transitions $(s,x) \trans{t} (s,x+t)$ before the actual input/output transition. 

\begin{definition}\label{def:Half-bisim}
Given a TFSM with timed guards and timeouts  $T = (S, I, O, \lambda_S, s_0, \Delta_S)$ and an untimed FSM $U = (R, I \cup \{\Half\}, O \cup \{\Half\}, \lambda_R, r_0)$, a \emph{$\Half$-bisimulation} is a relation $\sim \subseteq (S\times \bbR^+) \times R$ that respects the following conditions for every pair of states $(s,x) \in S\times \bbR^+$ and $r \in R$ such that $(s,x) \sim r$:
\begin{compactenum}
	\item if $(s,x)\trans{t}(s',x')$ with $0 < t < 1$ and either $x \in \mathbb{N}$ or $x + t \in \mathbb{N}$ then there exists $r'\in R$ such that $r \trans{\Half,\Half} r'$ and $(s',x') \sim r'$;
	\item if $r \trans{\Half,\Half} r'$ then for every $0 < t < 1$ such that either $x \in \mathbb{N}$ or $x + t \in \mathbb{N}$ there exists $(s',x')\in S\times \bbR^+$ such that $(s,x)\trans{t}(s',x')$ and $(s',x') \sim r'$;
\item if $(s,x)\trans{i,o}(s',0)$ then there exists $r'\in R$ such that $r \trans{i,o} r'$ and $(s',0) \sim r'$;	
	\item if $r \trans{i,o} r'$ then there exists $(s',0)\in S\times \bbR^+$ such that $(s,x)\trans{i,o}(s',0)$ and $(s',0) \sim r'$.
\end{compactenum}
$T$ and $U$ are  \emph{$\Half$-bisimilar} if there exists a $\Half$-bisimulation $\sim \subseteq S \times R$ such that $(s_0,0) \sim r_0$.
\end{definition}

The following lemma proves that $\Half$-bisimilar machines have the same behavior.

\begin{lemma}\label{lem:Half-bisimilar-then-equiv}
Given a TFSM with timeouts $T = (S, I, O, \lambda_S, s_0, \Delta_S)$ and an untimed FSM $U = (R, I\cup\{\Half\}, O\cup\{\Half\}, \lambda_R, r_0)$, if there exists a $\Half$-bisimulation $\sim$ such that $(s_0,0) \sim r_0$ then for every timed input word $v = (i_1,t_1)\dots(i_m,t_m)$ we have that $\Half(B_T(v)) = B_U(\Half(v))$.
\end{lemma}

\begin{proof}
The claim can be proved using the same argument of Lemma~\ref{lem:One-bisimilar-then-equiv}. See~\cite{techrep2014} for details.
\end{proof}

\begin{lemma}\label{lem:abstract-Half-bisim}
A TFSM with timeouts and timed guards $M$ is $\Half$-bisimilar to the abstract FSM $A_M$.
\end{lemma}

\begin{proof}
The relation $\sim = \{((s,x),(s,\langle n,n'\rangle))\mid x \in \langle n,n'\rangle \}$ is a $\Half$-bisimulation for $M$ and $A_M$.
\end{proof}

\begin{theorem}\label{th:tfsm-to-equiv}
Let $M$ and $M'$ be two TFSM with timeouts and timed guards.
Then $M$ and $M'$ are equivalent if and only if the two abstract FSM $A_M$ and $A_{M'}$ are equivalent.
\end{theorem}


\section{Comparison of TFSM models}
\label{sec:comparison}

In this section we compare the TFSM models considered in this paper with respect to their expressivity. 
It is easy to see that the class of TFSM with timeouts and timed guards includes both TFSM with timeouts and TFSM with timed guards. We start our comparison by showing that TFSM with timed guards and TFSM with timeouts are incomparable. To this end, consider the two TFSM of Figure~\ref{fig:counterexamples}: the following lemmas prove that there is no TFSM with timed guards equivalent to the TFSM with timeouts $M_1$, and no TFSM with timeouts equivalent to the TFSM with timed guards $M_2$.
 
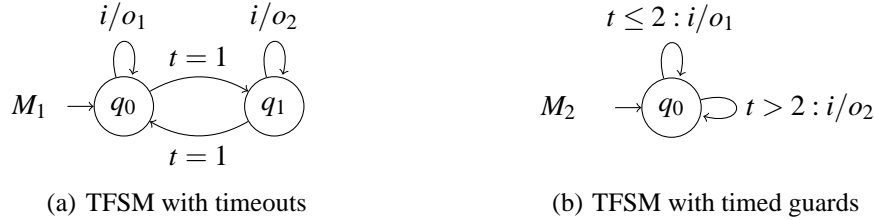
\begin{figure}[tbp]
\ \hfill
	\subfigure[TFSM with timeouts]{
	\begin{tikzpicture}
		\node	(M1) at (-1.25,0) {$M_1$};
		\node[draw,circle] (q0)	 at (0,0)	{$q_0$};
		\node[draw,circle]	(q1) at (2,0)	{$q_1$};
		
		\draw[->] (-0.75,0) -- (q0);
		\path[->] (q0) edge[loop above] node        {$i/o_1$} ();
		\path[->] (q0) edge[bend left]  node[above] {$t=1$} (q1);
		\path[->] (q1) edge[bend left]  node[below] {$t=1$} (q0);
		\path[->] (q1) edge[loop above] node        {$i/o_2$} ();
		\path (-1.75,-0.9)--(3.25,1.25);
	\end{tikzpicture}\label{fig:tfsm-to}}
\hfill
	\subfigure[TFSM with timed guards]{
	\begin{tikzpicture}
		\node	(M2) at (-1.5,0) {$M_2$};
		\node[draw,circle] (q0)	 at (0,0)	{$q_0$};

		\draw[->] (-0.75,0) -- (q0);
		\path[->] (q0) edge[loop above] node        {$t \leq 2: i/o_1$} ();
		\path[->] (q0) edge[loop right] node        {$t > 2: i/o_2$} ();
		\path (-2,-0.9)--(3,1.25);
	\end{tikzpicture}\label{fig:tfsm-tg}}
\hfill\ 
	\caption{Expressivity counterexamples.}
  \label{fig:counterexamples}
\end{figure}

\begin{proposition}
Let $I=\{i\}$, $O = \{o_1,o_2\}$ and consider the complete TFSM with timeouts $M_1$ depicted in Figure~\ref{fig:tfsm-to}. The behavior of $M_1$ cannot be described by any TFSM with timed guards.  
\end{proposition}

\begin{proof}
It is easy to see that $M_1$, under input $i$, produces the output $o_1$ only at those time instants $t$ such that $2n \leq t < 2n + 1$ for some natural $n$, while for time instants $t'$ such that $2n+1 \leq t' < 2n + 2$ the machine produces the output $o_2$.
Suppose that there exists a complete TFSM with timed guards $M_1' = (S, I, O, \lambda_S, s_0, \Delta_S)$ that is equivalent to $M_1$. Let $t_{max}$ be the maximum value that appears on the guards of the transitions exiting the initial state $s_0$ of $M_1'$ that are labelled with $i/o_1$. Two cases may arise.

\begin{itemize}
\item $t_{max} < +\infty$. In this case, let $n$ be such that $t_{max} < 2n + \frac{1}{2}$, and let $w$ be the single-letter timed input word $(i, 2n+\frac{1}{2})$. $M_1$ accepts $w$ and produces the corresponding timed output word $(o_1, 2n+\frac{1}{2})$. Consider now the TFSM with timed guards $M_1'$. At time $2n+\frac{1}{2}$ the TFSM $M_1'$ is still in the initial state $s_0$. However, no active transition exiting $s_0$ at time $2n+\frac{1}{2}$ is labelled with $i/o_1$, since $t_{max} < 2n + \frac{1}{2}$. Since $M_1'$ is complete, there must exist an active transition from $s_0$ at time $2n+\frac{1}{2}$, but this transition must be labelled with $i/o_2$. Hence, under input $w$, $M_1'$ must produce the timed word $(o_2, 2n+\frac{1}{2})$, in contradiction with the hypothesis that $M_1$ and $M_1'$ are equivalent.

\item $t_{max} = +\infty$. In this case, there is a transition exiting $s_0$ labelled with $i/o_1$ and with an interval $\langle t_f,+\infty)$, for some $t_f < +\infty$. Let $n$ be such that $t_{f} < 2n + 1 + \frac{1}{2}$, and let $w$ be the single-letter timed input word $(i, 2n+1+\frac{1}{2})$. $M_1$ accepts $w$ and produces the output word $(o_2, 2n+1+\frac{1}{2})$. Consider now the TFSM with timed guards $M_1'$. Since no input is received before, at time $2n+1+\frac{1}{2}$ the TFSM $M_1'$ is still in the initial state $s_0$. However, since  $t_{f} < 2n +1 +\frac{1}{2}<+\infty$, a transition labelled with $i/o_1$ is active at time $2n +1 +\frac{1}{2}$. Hence, $M_1'$ accepts $w$ and produces the timed output word $(o_1,2n+1+\frac{1}{2})$.
This is a contradiction with the hypothesis that $M_1$ and $M_1'$ are equivalent. \qedhere
\end{itemize}
\end{proof}

\begin{proposition}
Let $I=\{i\}$, $O = \{o_1,o_2\}$ and consider the TFSM with timed guards $M_2$ depicted in Figure~\ref{fig:tfsm-tg}.  
The behavior of $M_2$ cannot be described by any TFSM with timeouts.  
\end{proposition}

\begin{proof}
It is easy to see that $M_2$ (under input $i$) produces the output $o_1$ only at those time instants $t$ such that $t \leq 2$, while for time instants strictly greater than $2$ the machine  produces the output $o_2$.
Suppose by contradiction that there exists a TFSM with timeouts $M_2' = (S, I, O, \lambda_S, s_0, \Delta_S)$ that is equivalent to $M_2$, and consider the single-letter timed input word $w = (i, 2)$. We have that $M_2$ accepts $w$ and produces the output $o_1$ at time $2$. Since $M_2'$ is equivalent to $M_2$, we have that $w$ must be accepted by $M_2'$ as well. Let $(s,x)$ be a timed state of $M_2'$ such that $(s_0,0) \trans{2} (s,x)$. Since $w$ is accepted, we have that there must exist a transition $(s, i, o_1, s') \in \lambda_S$ for some $s' \in S$. Now, let $w'$ be the timed input word $(i, 2 + \varepsilon)$, for some $\varepsilon < 1 \leq \Delta_S(s)_{\downarrow \bbN}$. By the definition of TFSM with timeouts, we have that $(s_0,0) \trans{2 + \varepsilon} (s,x+\varepsilon)$; thus $w'$ is accepted by $M_2'$, which produces the output $o_1$ at time $2 + \varepsilon$, since the input/output transition $(s, i, o_1, s') \in \lambda_S$ can be triggered from the timed state $(s,x+\varepsilon)$. However, $M_2$ produces the output $o_2$ with input $w'$, in contradiction with the hypothesis that $M_2$ and $M_2'$ are equivalent.
\end{proof}

The above examples show that, in general, TFSM with timeouts cannot be transformed into TFSM with timed guards, and that TFSM with timed guards cannot be transformed into TFSM with timeouts. We will now study the restrictions under which this transformation is possible, 

Consider the TFSM with timeouts $M_1$ of Figure~\ref{fig:tfsm-to}: it shows that that cycles of timeout transitions cannot be captured by a TFSM with timed guards. However, when there are no such cycles, we can use Algorithm~\ref{alg:fsm-lf-to-fsm-tg} to transform a TFSM with timeouts into a TFSM with timed guards.

\begin{algorithm}[tbp] 
\caption{Transform a loop-free TFSM with timeouts into a TFSM with timed guards}
\label{alg:fsm-lf-to-fsm-tg}
\KwIn{A loop-free TFSM with timeouts $M =(S, I, O, \lambda_S, s_0, \Delta_S)$}
\BlankLine
	\While{there exists $s_j \in S$ such that $\Delta_S(s_j)_{\downarrow \bbN} < \infty$}{
		\Let{$s_j \in S$ such that $\Delta_S(s_j)_{\downarrow \bbN} < \infty$}\;
		\Let{$(s_k,n) = \Delta_S(s_j)$}\;
		\lForEach{$(s_k,i,[t_1,t_2),o,s_h) \in \lambda_S$}{
			$\lambda_S = \lambda_S \cup \{(s_j,i,[t_1+n,t_2+n),o,s_h)\}$\;
		}
		\lIf{$\Delta_S(s_k)_{\downarrow \bbN} < \infty$}
			{$\Delta_S(s_j) = (\Delta_S(s_k)_{\downarrow S},\Delta_S(s_k)_{\downarrow \bbN}+n)$}
		\lElse{$\Delta_S(s_j) = (s_j, \infty)$}\;
	}
	\Return($S, I, O, \lambda_S, s_0$)\;
\end{algorithm}

\begin{proposition}\label{prop:fsm-lf}
Given a TFSM with timeouts $M$ without loops of timeout transitions, Algorithm~\ref{alg:fsm-lf-to-fsm-tg} terminates and builds a TFSM with timed guards that is equivalent to $M$. 
\end{proposition}

\begin{proof}
See ~\cite{techrep2014}.
\end{proof}

Consider now the TFSM with timed guards $M_2$ of Figure~\ref{fig:tfsm-tg}: it suggests that transitions with left-open guards $( t_1, t_2\rangle$  or with right-closed guards $\langle t_1, t_2 ]$  cannot be captured by TFSM with timeouts. The following proposition shows that if we force all guards to be left-closed and right-open intervals $[t_1, t_2)$, then the translation into a TFSM with timeouts is possible.

\begin{proposition}\label{prop:fsm-lcro}
The behavior of a TFSM with timed guards can be described by a TFSM with timeouts if all guards are left-closed and right-open.
\end{proposition}

\begin{proof}
Let $M =(S, I, O, \lambda_S, s_0)$ be a TFSM with timed guards where every transition $(s,i,g,o,s') \in \lambda_S$ is such that $g = [t_{min},t_{max})$ is a left-closed and right-open interval. We build a TFSM with timeouts $M' = (S', I, O, \lambda_S',s_0',\Delta_S')$ equivalent to $M$ as follows.

\begin{enumerate}
\item Given an enumeration of the states $S = \{s_0, s_1, \dots, s_n\}$ of $M$, for each state $s_j \in S$ we define the corresponding set $T_j = \{0 = t_0^j < t_1^j < \dots < t_m^j = \infty\}$ as the set of all values that appear in the guards of transitions departing from $s_j$. The states of $M'$ are pairs $(s_j,[t_k^j,t_{k+1}^j))$ where $s_j \in S$ is a state of $M$ and $t_k^j,t_{k+1}^j \in T_j$. Intuitively, a state $(s_j,[t_k^j,t_{k+1}^j))$ of $M'$ corresponds to  the situation when $M$ is in state $s_j$ with clock value in the interval $[t_k^j,t_{k+1}^j)$. 
\item The timeout transition is defined as $\Delta_S'((s_j,[t_k^j,t_{k+1}^j))) = ((s_j,[t_{k+1}^j,t_{k+2}^j)),$ $t_{k+1}^j-t_k^j)$ when $t_{k+1}^j < \infty$, and as $\Delta_S'((s_j,[t_k^j,t_{k+1}^j))) = ((s_j,[t_{k}^j,t_{k+1}^j)),\infty)$ when $t_{k+1}^j = \infty$. 
\item For every  $(s_j,i,[t_{min},t_{max}),o,s_k) \in \lambda_S$ we put a transition $((s_j,[t_h^j,t_{h+1}^j)),i,o,(s_k,[0,t_1^k))) \in \lambda_S'$ for every $h$ such that $t_{min} \leq t_h^j$ and $t_{h+1}^j\leq t_{max}$.
\item The initial state of $M'$ is $s_0' = (s_0,[0,t_1^0))$.
\end{enumerate}

To prove that $M$ and $M'$ have the same behavior, we have to show that for every sequence of transitions $(s_j,0) \trans{t} (s_j, t) \trans{i,o} (s_k,0)$ of $M$ there is a corresponding sequence in $M'$, and vice-versa. Now, since $(s_j,t) \trans{i,o} (s_k,0)$ is an input/output transition of $M$, 
we have that there exists a transition $(s_j,i,[t_{min},t_{max}),o,s_k) \in \lambda_S$ such that $t \in [t_{min},t_{max})$.
By 2., we have that $((s_j,[0,t_1^j)),0) \trans{t} ((s_j,[t_h^j,t_{h+1}^j)),t')$ for some $h$ such that $t_h^j \leq t < t_{h+1}^j$, and with $t' = t - t_h^j$. By 3., we have that $((s_j,[t_h^j,t_{h+1}^j)),i,o,(s_k,[0,t_{1}^k))) \in \lambda_S'$. Hence, $((s_j,[0,t_1^j)),0) \trans{t} ((s_j,[t_h^j,t_{h+1}^j)),t') \trans{i,o} ((s_k,[0,t_1^k)),0)$ is a sequence of transitions of $M'$.

To conclude the proof, suppose that $((s_j,[0,t_1^j)),0) \trans{t} ((s_j,[t_h^j,t_{h+1}^j)),t') \trans{i,o} ((s_k,[0,t_1^k)),0)$ is a sequence of transitions of $M'$. By 2., we have that $h$ is such that $t_h^j \leq t < t_{h+1}^j$. Since $((s_j,[t_h^j,t_{h+1}^j)),t') \trans{i,o} ((s_k,[0,t_1^k)),0)$, we have that there exists a transition $((s_j,[t_h^j,t_{h+1}^j)),i,o,(s_k,[0,t_1^k))) \in \lambda_S'$. By 3., there exists a transition $(s_j,i,[t_{min},t_{max}),o,s_k) \in \lambda_S$ such that $t_{min} \leq t_h^j$ and $t_{h+1}^j \leq t_{max}$. Hence, we have proved that $(s_j, 0) \trans{t} (s_j,t) \trans{i,o} (s_k,0)$ is a valid sequence of transitions in $M$.  \end{proof}

The results obtained in this section are summarized in Figure~\ref{fig:comparison}: neither TFSM with timed guards nor TFSM with timeouts are sufficient to describe the behavior of each other, nor of TFSM with both timed guards and timeouts. Moreover, we would like to point out that Algorithm~\ref{alg:fsm-lf-to-fsm-tg} builds a TFSM with LCRO guards from a loop-free TFSM with timeouts, and that the construction in the proof of Proposition~\ref{prop:fsm-lcro} builds a loop-free TFSM with timeouts from a TFSM with LCRO guards. This allows us to conclude that the two classes are equivalent.

\begin{figure}[tbp]
   \centering
	\begin{tikzpicture}[xscale=1.5,font=\footnotesize]
		\node (tfsm) at (0,0) 	{TFSM with timed guards and timeouts};
		\node (fsm-to)	at (-2,-0.9) {TFSM with timeouts};
		\node (fsm-tg)	at (2,-0.9) {TFSM with timed guards};
		\node (fsm-lf) at (-2,-2.1) {Loop-free TFSM with timeouts};
		\node (fsm-lcro) at (2,-2.1) {TFSM with LCRO timed guards};
		\node (fsm) at (0,-3) {Untimed FSM};
		\draw[->] (tfsm) -- (fsm-to);
		\draw[->] (tfsm) -- (fsm-tg);
		\draw[->] (fsm-to) -- (fsm-lf);
		\draw[->] (fsm-to) -- (fsm-lcro);
		\draw[->] (fsm-tg) -- (fsm-lf);
		\draw[->] (fsm-tg) -- (fsm-lcro);
		\draw[<->] (fsm-lcro) -- (fsm-lf);
		\draw[->] (fsm-lcro) -- (fsm);
		\draw[->] (fsm-lf) -- (fsm);
	\end{tikzpicture}

   \caption{Comparison of TFSM models.}\label{fig:comparison}
\end{figure}
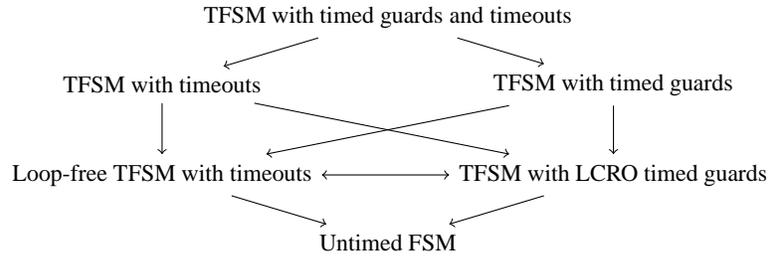

\section{Conclusions}
\label{sec:conclusions}

In this paper, we investigated some models of deterministic TFSMs with 
a single clock: TFSMs with only timed guards, TFSMs with only timeouts, 
and TFSMs with both timed guards and timeouts.
We showed that the behaviours of the timed FSMs are equivalent if and only 
if the behaviours of the companion untimed FSMs obtained by time-abstracting
bisimulations are equivalent.
Moreover, we compared their expressive power, characterizing subclasses 
of TFSMs with timed guards and of TFSMs with timeouts that are equivalent 
to each other.
These timed FSM models exhibit a good trade-off between
expressive power and ease of analysis.
We are currently generalizing these results to non-deterministic timed FSMs,
and comparing our models with classical FSMs and special classes of timed 
automata (e.g., with a single clock).
Future work includes deriving tests for a timed FSM with timed guards and
timeouts, extending the derivation for a timed FSM with timed guards and for
a timed FSM with timeouts, respectively, in~\cite{ElFakih-scp2014}
and in~\cite{Zhigulin2011}.
Finally, we will define the composition of timed FSMs and investigate
the solution of equations over timed FSMs to synthesize unknown
timed FSMs~\cite{villa-ucp-book}.

\bibliographystyle{eptcs}
\bibliography{tfsms}

\end{document}